\documentclass[11pt]{article}

\usepackage{amsmath,amssymb,amsthm}
\usepackage{amsfonts}
\usepackage{multicol}

\usepackage{color}

\newtheorem{theorem}{Theorem}
\newtheorem{lemma}[theorem]{Lemma}
\newtheorem{prop}[theorem]{Proposition}
\newtheorem{corollary}[theorem]{Corollary}
\newtheorem{e-definition}[theorem]{Definition\rm}
\newtheorem{remark}{\bf Remark\/}

\newenvironment{definition}[1][Definition]{\begin{trivlist}
\item[\hskip \labelsep {\bfseries #1}]}{\end{trivlist}}

\usepackage{color}

\title{Market-consistent valuation of insurance liabilities \\ by cost of capital}
\author{Christoph M\"ohr}
\date{\today}

\begin{document}

\maketitle

\begin{abstract}
This paper investigates market-consistent valuation of insurance liabilities in the context of, for instance, Solvency II and to some extent IFRS 4. We propose an explicit and consistent framework for the valuation of insurance liabilities which incorporates the Solvency II approach as a special case.

The proposed framework is based on dynamic replication over multiple (one-year) time periods by a portfolio of assets with reliable market prices, allowing for "limited liability" in the sense that the replication can in general not always be continued. The asset portfolio consist of two parts: (1) assets whose market price defines the value of the insurance liabilities, and (2) capital funds used to cover risk which cannot be replicated. The capital funds give rise to capital costs; the main exogenous input of the framework is the condition on when the investment of the capital funds is acceptable.

We investigate existence of the value and show that the exact calculation of the value has to be done recursively backwards in time, starting at the end of the lifetime of the insurance liabilities. The main question only partially considered in this paper is the uniqueness of the value. We derive upper bounds on the value and, for the special case of replication by risk-free one-year zero-coupon bonds, explicit recursive formulas for calculating the value.

Valuation in Solvency II and IFRS 4 is based on representing the value as a sum of a "best estimate" and a "risk margin". In our framework, it turns out that this split is not natural. Nonetheless, we show that a split can be constructed as a simplification, and that it provides an upper bound on the value under suitable conditions. We illustrate the general results by explicitly calculating the value for a simple example.
\end{abstract}

\vspace{10pt}


\vspace{10pt}

{\bf Acknowledgments.} I owe my gratitude to Dr. Philipp Keller, who has suggested and supported the investigation of this topic.

\newpage


\section*{Introduction}

Our starting point is market-consistent valuation of insurance liabilities ("technical provisions") under Solvency II. References to the approach include the Solvency II Framework Directive, DIRECTIVE 2009/138/EC \cite{solv2}, the draft Level 2 Implementation Measures, "Rules relating to technical provisions", EIOPC/SEG/IM13/2010 \cite{im13}, as well as related documents such as the CRO Forum position paper \cite{croforummvl}, the report by the Risk Margin Working Group \cite{iaamlc}, and CEIOPS-DOC-36/09 (former CP 42) \cite{ceiops42}. Many of the concepts used by Solvency II had earlier been introduced in the Swiss Solvency Test (SST), see for instance Federal Office of Private Insurance \cite{sst}.

In Solvency II, according to Article 77 in DIRECTIVE 2009/138/EC \cite{solv2}, the market-consistent value of an insurance liability is determined in one of two ways: If the cash-flows of the liability (or part of the cash-flows) can be replicated reliably using financial instruments for which a reliable market value is observable, then the value (of the part of the cash-flows) is determined on the basis of the market value of these instruments. Otherwise, the value is equal to the sum of best estimate and risk margin,
\begin{equation}\label{def-mcvl}
\mbox{Market-consistent value} = \mbox{best estimate}+\mbox{risk margin}.
\end{equation}
In Article 77 of the DIRECTIVE 2009/138/EC \cite{solv2}, the best estimate is defined as the "probability-weighted average of future cash-flows, taking account of the time value of money (expected present value of future cash-flows), using the relevant risk-free interest rate term structure," and the risk margin is "calculated by determining the cost of providing an amount of eligible own funds equal to the Solvency Capital Requirement necessary to support the insurance and reinsurance obligations over the lifetime thereof."

In EIOPC/SEG/IM13/2010 \cite{im13}, the risk margin is expressed by a cost of capital approach as the sum of the costs of future required capital $SCR_t$ by the expression
\begin{equation}\label{def-riskmargininitial}
\mbox{Risk margin} = CoC \cdot \sum_{t\geq 0} \frac{SCR_t}{(1+r_{t+1})^{t+1}}.
\end{equation}
where $CoC$ denotes the cost of capital rate, which is assumed deterministic and constant and, in EIOPC/SEG/IM13/2010 \cite{im13}, is set to $6\%$ above the risk-free rate. The sum is over all years $t$, and $r_{t+1}$ denotes the risk-free discount rate for $t+1$ years, which means that the capital costs of year $t$ are discounted back from the end of year $t$. The infinite sum above will be finite in practice, limited by the lifetime of the corresponding liabilities.

In the formula \eqref{def-riskmargininitial}, $SCR_t$ denotes the Solvency Capital Requirement from Solvency II for the year $t$, i.e. the required capital, which is defined in Article 101 of DIRECTIVE 2009/138/EC \cite{solv2} to correspond "to the Value-at-Risk of the basic own funds of an insurance or reinsurance undertaking subject to a confidence level of $99,5\%$ over a one-year period." We consider $SCR_t$ in more detail in Section~\ref{sec-s2}, but note here the following: for future years $t>0$, $SCR_t$ depends on the future state at the beginning of year $t$, which is currently not known. Consequently, $SCR_t$ for $t>0$ is a random variable, implying that the risk margin as defined by \eqref{def-riskmargininitial} is a random variable and not a number, as it ought to be.

To avoid this problem, every $SCR_t$ in \eqref{def-riskmargininitial} could be replaced with the current expected value of the random variable $SCR_t$, so that the risk margin would correspond to the \emph{expected} costs of future required capital. On might then think that this "expected risk margin" is only sufficient in expectation. However, as we show in Section~\ref{sec-rm}, it turns out that, under suitable assumptions, the "expected risk margin" is sufficient not just in expectation but always.

As an additional complication, according to EIOPC/SEG/IM13/2010 \cite{im13}, the $SCR_t$ used for calculating the risk margin in Solvency II is not calculated for the company (undertaking) under consideration, but for a "reference undertaking" to which the insurance liabilities are hypothetically transferred. The features of the transfer and the properties of the reference undertaking are specified in detail in EIOPC/SEG/IM13/2010 \cite{im13}.

The preceding comments aim to indicate that valuation of insurance liabilities according to Solvency II is not obvious and that a more explicit framework might be needed. The objective of this paper is to propose such a framework, which incorporates the Solvency II approach as a special case. The proposed framework expresses the value in terms of the market price of a portfolio of assets. It is based on dynamic replication over multiple time periods of the cash-flows of the insurance liabilities by portfolios of assets with reliable market prices. In this sense, it relies on the seminal idea of valuation by replication underlying the (risk-neutral) pricing of financial derivatives.

The framework needs to capture two additional aspects. The first additional aspect is that insurance liabilities can typically not be perfectly replicated by assets with reliable market prices, so there remains a part of the cash-flows which cannot be replicated. According to Solvency II, the non-replicable part of the cash-flows is covered by capital funds, giving rise to capital costs. The second additional aspect is that the replication cannot always be continued. In Solvency II, this is because the required capital funds are given by the Solvency Capital Requirement in terms of the VaR at $99.5\%$, which implies that they will be insufficient with $0.5\%$ probability.

The main exogenous assumption in the framework is what we call the "acceptability condition" in the remainder of this paper. The acceptability condition is the condition on when the stochastic return on the capital funds is acceptable to the investor of the capital funds. In other words, it specifies the "price" of the capital investment. In this paper, we work with the acceptability condition implicit in the definition of the risk margin in Solvency II, which is that the expected excess return over the risk-free return is equal to the capital cost rate $CoC$. We note that this acceptability condition is formulated independently of the capital investor and so does not take into account the specific risk profile of a given investor.

In general, the value of the insurance liabilities can depend on the assumptions made about future new business written, as future new business might diversify with the run-off of the current business. In this paper, we consider a "run-off" situation in the sense that we assume that no future new business is written.

Under the proposed framework, it turns out that a precise calculation of the value needs to be done recursively backwards in time, starting at the end of the lifetime of the insurance liabilities. Moreover, we find that there is no natural split of the value into a "best estimate" and a "risk margin"; the value is simply given as the market price of a specific portfolio of assets. However, we show in Section~\ref{sec-rm} that, under certain conditions, a split can be introduced, and that the resulting sum of "best estimate" and "risk margin" is not equal to the value but provides an upper bound.

The proposed framework can be situated in the context of (market-consistent) valuation in incomplete markets. At present, on the one hand, there is extensive academic literature on aspects of valuation by dynamic replication and in incomplete markets, while, on the other hand, from a practitioner's perspective, there are numerous articles about certain aspects of the Solvency II valuation, such as simplified approaches, the risk-free rate, the cost of capital rate etc. This paper aims to bridge the two areas, by formulating Solvency II valuation in the framework of dynamical replication in incomplete markets and in this way defining what it really means.

The recent paper Salzmann-W\"uthrich \cite{sawue} provides a discussion of a mathematically consistent multi-period risk measure approach for the calculation of a risk margin to cover possible shortfalls in the liability runoff of general (i.e. non-life) insurance companies. Moreover, explicit calculations are presented by means of a Bayes chain ladder model and a risk measure chosen to be a multiple of the standard deviation.

Our approach is related to the Valuation Portfolio (VaPo) according to B\"uhlmann \cite{vapo} and W\"uthrich et al. \cite{wbf}: An insurance obligation can be better understood not in terms of monetary values but as a collection of appropriately chosen financial instruments. In contrast to the VaPo approach, we do not express the actual liability as a portfolio of potentially synthetic instruments, but consider replication of the liability's cash-flows by a portfolio of assets with reliable market prices.

The risk margin in the context of the one-year risk is also investigated in Ohlsson-Lauzeninsks\cite{ohllau}.

We mention here also the classical paper Artzner et al. \cite{cohrm} on coherent risk measures or, equivalently, "acceptable future net worths". While risk measures play a prominent part in what follows, that paper considers a one-period setting and does not consider dynamic replication.

An alternative approach to the acceptability condition is given by utility indifference pricing similar to M{\o}ller \cite{actvalprinc}.

\subsection*{Organization of the paper}

The paper is organized as follows: In Section~\ref{sec-setup}, we set up the mathematical notations and assumptions, including the filtration used to express available information and risk-free discounting. In Section~\ref{sec-s2}, we investigate the Solvency II approach to valuation and solvency as a motivation for our formulation, in Section~\ref{sec-frame}, of the proposed framework for valuation. In Section~\ref{sec-oneyear}, we then investigate valuation over multiple one-year time periods in the proposed framework. In Section~\ref{sec-rm}, we consider the risk margin and prove one of the main results of this paper: under suitable assumptions, the sum of "best estimate" and "risk margin" is an upper bound for the value. Finally, in Section~\ref{sec-ex}, we explicitly calculate the value for a simple example and show that the upper bound sometimes reverses the "ordering" of the value between different liabilities.

\section{Set up and notation}\label{sec-setup}

We consider time periods of one year, where year $t=0,1\ldots$ refers to the time-period $[t,t+1)$. To be able to describe actions taken at the end of year $t$, we denote  by $(t+1)^-$ a point in time just before time $t+1$.

We assume that there exists a filtration $(\mathcal{F}_t)_t$ expressing the information available (known) at time $t$. To specify the filtration, we use the notation identical to W\"uthrich et al. \cite{wbf}. That is, we define a filtered probability space by choosing a probability space $(\Omega, \mathcal{F}, \mathcal{P})$ and an increasing sequence of $\sigma$-fields $(\mathcal{F}_t)_{t=0,\ldots,n}$ with
\begin{equation*}
\{ \emptyset, \Omega\}=\mathcal{F}_0\subseteq \mathcal{F}_1 \subseteq\ldots \subseteq\mathcal{F}_n
\end{equation*}
where we assume $\mathcal{F}_n =\mathcal{F}$ for simplicity. All random variables considered are assumed to be adapted to the filtration $(\mathcal{F}_t)_t$.

Given a set $A$, we denote its complement by $A^c$ and its characteristic function by $1_A$. The characteristic function takes the value $1$ on $A$ and $0$ on $A^c$. We consider risk measures $\rho$, taking a random variable $X$ to a real number $\rho\{ X\}$. We define losses to be negative numbers and the risk $\rho$ of a loss to be a positive number. A risk measure $\rho$ is called \emph{translation-invariant} (or cash-invariant) if, for any random variable $X$ and any real number $b$, $\rho\{ X+b\}=\rho\{ X\}-b$. It is called \emph{monotone} if, for any two random variables with $X_1\leq X_2$, we have $\rho\{ X_2\}\leq \rho\{ X_1\}$.

The main objective of the paper is the valuation of a given insurance liability $\mathcal{L}$ with stochastic cash-flows $(X_t)_t$ corresponding to claims payments, expenses etc., where $X_t$ denotes the cash-flow in year $t$. For simplicity, we assume that the cash-flow $X_t$ occurs and is known at time $(t+1)^-$. In terms of the filtered probability space, we assume that the cash-flows $(X_t)_t$ are adapted to the filtration $(\mathcal{F}_t)_t$ and that $X_t$ is $\mathcal{F}_{t+1}$-measurable. At time $t$, for instance, intuitively speaking, the value of $X_t$ is not known, but the distribution of $X_t$ is known.

We assume throughout the paper that market prices of certain financial instruments are available at future points in time. That is, the information $\mathcal{F}_t$ available at time $t$ includes the market prices of financial instruments at time $t$. The future market prices of instruments are given by stochastic models.

A \emph{reference market} (or replicating market) is defined to be a set of financial instruments for which reliable market prices are assumed to exist. As an idealization, market prices of an instrument are reliable if any quantity of the instrument can instantaneously be traded without affecting the market price. Typically, it is assumed that, if an instrument is traded in a deep and liquid market, then its (unique and additive) reliable market price is an emergent property of the corresponding market. An asset portfolio consisting of instruments from the reference market is called a \emph{reference portfolio} (or replicating portfolio).

Deep and liquid (and transparent) markets are defined in the Solvency II context in EIOPC/SEG/IM13/2010 \cite{im13}, which also specifies that the model used for the projection of market parameters (or market prices) needs to ensure that no arbitrage opportunity exists. In line with this requirement, we assume in the following that the reference market is arbitrage-free.

We assume that risk-free zero-coupon bonds are part of the reference market, and do not specify which other instruments might be in the reference market. As mentioned above, we assume models for the stochastic future market prices for the instruments in the reference market.

To express risk-free discounting of a cash-flow $x$ occurring at time $s$ discounted to time $t\le s$, we write
\begin{equation*}
\text{pv}_{(s\to t)}(x)
\end{equation*}
which is to be understood as the value at time $t$ of a risk-free zero-coupon bond in the appropriate currency with face value $x$ maturing at time $s$. It is in this sense not possible to risk-free discount stochastic (as opposed to deterministic) cash-flows, because the cash-flow of a risk-free zero-coupon bond is deterministic.

We define the risk-free terminal value of an amount $x$ invested at time $t$ in a risk-free zero-coupon bond maturing at time $s\ge t$ by
\begin{equation*}
\text{tv}_{(t\to s)}(x).
\end{equation*}
Let $R_t^{(m)}$ denote the annual rate for a risk-free zero-coupon bond at time $t$ with a term of $m=1,2\ldots$ years, so $R_t^{(m)}$ is $\mathcal{F}_t$-measurable, and
\begin{equation}\label{def-annualriskfreerate}
\text{pv}_{(t+m\to t)}(x_{t+m})=( 1+R_t^{(m)})^{-m}\cdot x_{t+m}.
\end{equation}
Consider a risk-free forward contract set up at time $t$, which specifies that, at time $t+1$, for a price of  $B_{t+1}^m(t)$ fixed at time $t$, a risk-free zero-coupon bond is purchased with a payoff of $1$ at time $t+1+m$. Because of no-arbitrage, we must have that
\begin{equation}\label{eq-noarbitrageforward}
( 1+R_{t}^{(1)})^{-1}\cdot B_{t+1}^m(t)=( 1+R_{t}^{(m+1)})^{-m-1}.
\end{equation}
It is common to identify the forward price with the expectation at time $t$ of the corresponding bond price, i.e.
\begin{equation*}
B_{t+1}^m(t)=\mathbb{E}\{ ( 1+R_{t+1}^{(m)})^{-m}  \mid\mathcal{F}_t\}.
\end{equation*}
In general, the price of a forward contract might contain an additional premium for liquidity, so
\begin{equation}\label{assumpt-forwardprice2}
B_{t+1}^m(t)\geq \mathbb{E}\{ ( 1+R_{t+1}^{(m)})^{-m}  \mid\mathcal{F}_t\}.
\end{equation}
Equations \eqref{assumpt-forwardprice2} and \eqref{eq-noarbitrageforward} imply that
\begin{equation*}
( 1+R_{t}^{(1)})^{-1}\cdot \mathbb{E}\{ ( 1+R_{t+1}^{(m)})^{-m}  \mid\mathcal{F}_t\}\leq ( 1+R_{t}^{(m+1)})^{-m-1}.
\end{equation*}

\section{Market-consistent valuation in Solvency II}\label{sec-s2}

Because we are proposing a framework for valuation which incorporates Solvency II as a special case, we investigate in the following the Solvency II approach to valuation and solvency in more detail. The expressions we derive here are used to motivate the definition of the general framework in Sections~\ref{sec-frame} and \ref{sec-oneyear}.

To begin with, we consider the Solvency Capital Requirement $SCR_t$, which is defined to correspond "to the Value-at-Risk of the basic own funds of an insurance or reinsurance undertaking subject to a confidence level of $99,5\%$ over a one-year period." (DIRECTIVE 2009/138/EC \cite{solv2}).

For the actual balance sheet of the company (or "undertaking") under consideration, for simplicity, we identify in the following basic and eligible own funds (as defined under Solvency II) with the available capital, denoted by $AC_t$ at time $t$, which is defined as the difference between the market-consistent value $V_t(\mathcal{A}_t)$ of the assets $\mathcal{A}_t$ and the market-consistent value $V_t(\mathcal{L}_t)$ of the liabilities $\mathcal{L}_t$,
\begin{equation*}
AC_t:=V_t(\mathcal{A}_t)-V_t(\mathcal{L}_t).
\end{equation*}
$SCR_t$ can then be written in terms of the one-year change of the available capital,
\begin{equation}\label{def-scrt0}
SCR_t:=\text{pv}_{(t+1\to t)}\left(\rho\left\{ AC_{(t+1)^{-}}-\text{tv}_{(t\to t+1)}\left(AC_{t}\right)\mid\mathcal{F}_t\right\}\right),
\end{equation}
where the risk measure $\rho$ is prescribed to be the Value-at-Risk $VaR_{\alpha}$ at the $\alpha =99.5$-percentile
\begin{equation}\label{def-rhosii}
\rho \{Z\}:=VaR_{\alpha}\{-Z\}.
\end{equation}
Note that $SCR_t$ is calculated based on the information $\mathcal{F}_t$ available at time $t$.

$SCR_t$ is the capital requirement under Solvency II in the assessment of the solvency of a company. Solvency is effectively specified by the condition that, with $99.5\%$ probability, at the end of year $0$ (at time $t=1^{-}$), the market-consistent value of the assets exceed the market-consistent value of the liabilities,
\begin{equation*}
V_1(\mathcal{A}_{1^{-}})\geq V_1(\mathcal{L}).
\end{equation*}
which corresponds to the requirement at time $t=0$ that the available capital exceed the required capital,
\begin{equation*}
AC_0\geq SCR_0,
\end{equation*}
with $SCR_0$ given by \eqref{def-scrt0} for $t=0$. In order to assess the solvency condition, we in particular need to know the value of the insurance liabilities.

\vspace{10pt}

Regarding the value of the insurance liabilities, we recall from the introduction that the risk margin as a component of the value is defined in terms of the Solvency Capital Requirement $SCR_t$. However, $SCR_t$ is not calculated for the company which currently holds the insurance liabilities, but for a so-called reference undertaking to which the insurance liabilities $\mathcal{L}$ are hypothetically transferred for the purpose of valuation.

The features of this transfer and the properties of the reference undertaking are defined in EIOPC/SEG/IM13/2010 \cite{im13}. After the transfer, the liability side of the balance sheet of the reference undertaking is assumed to consist of the transferred insurance liabilities. The assets are assumed to consist of two parts. The first part is a reference portfolio of assets we denote by $RP_t$, which is used to cover the value of the insurance liabilities. That is, the value $V_t(\mathcal{L})$ of the insurance liabilities $\mathcal{L}$ at time $t$ is given by the market price of the reference portfolio
\begin{equation}\label{def-value}
V_t(\mathcal{L})=V_t(RP_t)
\end{equation}
The second part of the assets consists of available capital $AC_t$, assumed invested risk-free, equal to the Solvency Capital Requirement $SCR_t$ needed for the reference undertaking.

Under these specifications, $SCR_t$ from \eqref{def-scrt0} can be rewritten as follows. The available capital $AC_{(t+1)^-}$ at the end of year $t$ (before any potential recapitalization) is given by the year-end value of the assets reduced by the cash-flow $X_t$ in year $t$ and the year-end value of the insurance liabilities $V_{t+1}(\mathcal{L})$, i.e.
\begin{equation*}
AC_{(t+1)^-}=\text{tv}_{(t\to t+1)}\left(SCR_t\right)+V_{t+1}(RP_t)-X_t-V_{t+1}(\mathcal{L}).
\end{equation*}
Since $AC_t=SCR_t$, we get from \eqref{def-scrt0} the following formula for the $SCR_t$ for the purpose of valuation,
\begin{equation}\label{expr-scrt0}
SCR_t=\text{pv}_{(t+1\to t)}\left(\rho\left\{ V_{t+1}(RP_t)-X_t-V_{t+1}(\mathcal{L})\mid\mathcal{F}_t\right\}\right).
\end{equation}
It becomes clear from this expression that, in order to calculate the market-consistent value $V_{t}(\mathcal{L})$ at time $t$ by \eqref{def-value}, which through the risk margin (or through the acceptability condition \eqref{cond-expvalueinit} below) depends on $SCR_t$, one first needs to calculate the market-consistent value $V_{t+1}(\mathcal{L})$ at time $t+1$ etc. This implies that a precise calculation of the market-consistent value has to be recursively backwards in time.

The expression \eqref{expr-scrt0} also shows that underlying market-consistent valuation of insurance liabilities is dynamic replication with a one-year time period. At time $t$, the portfolio $RP_t$, which defines the value $V_{t}(\mathcal{L})$ through \eqref{def-value}, is set up to replicate the random variable $X_t+V_{t+1}(\mathcal{L})$ at time $t+1$. In the case of perfect replication, $V_{t+1}(RP_t)$ is always equal to $X_t+V_{t+1}(\mathcal{L})$, so that, at time $t+1$, a new replicating portfolio $RP_{t+1}$ can be constructed by a suitable reinvestment of the assets $RP_t$, and no capital funds are needed.

For insurance liabilities, perfect replication is typically not possible. Hence, additional capital funds are needed for the instances in which $V_{t+1}(RP_t)$ is less than the sum $X_t+V_{t+1}(\mathcal{L})$, so capital funds account for the part of the liability which cannot be replicated. This gives rise to capital requirements $SCR_t$ according to \eqref{expr-scrt0}, which depend on the real-world probabilities of different amounts of the difference $V_{t+1}(RP_t)-X_t-V_{t+1}(\mathcal{L})$. In general, future new business might be written and thus be added to the balance sheet in the future, and the corresponding cash-flows might diversify with the cash-flows of the liability $\mathcal{L}$ under consideration. Since insurance liabilities typically run-off over several years, this means that the current value of an insurance liability is potentially affected by insurance obligations which are added to the balance sheet in the future, i.e. future new business, at least until the liability is fully run-off.

In Solvency II, the assumptions on future new business in the calculation of the risk margin are currently not really clear. In this paper, we consider a "run-off" situation in the sense that we assume that no future new business is written.

The capital $SCR_t$ comes with a cost to make the capital investment acceptable to the capital provider, which we express through the acceptability condition. The acceptability condition is encoded in the definition of the risk margin, and requires that the expected return on the capital $SCR_t$ at the end of year $t$ be equal to a cost of capital rate $CoC$ in excess of the risk-free rate. The value of the capital investment at the end of the year is determined from the available capital $AC_{(t+1)^-}$, considering that its value is never negative, since the capital provider has limited liability. Hence, the acceptability condition for year $t$ can be written as
\begin{equation}\label{cond-expvalueinit}
\mathbb{E}\left\{\max\left\{0, AC_{(t+1)^-}\right\}\mid\mathcal{F}_t\right\}  =  \text{tv}_{(t\to t+1)}\left(SCR_t\right)+CoC\cdot SCR_t.
\end{equation}
The left hand side of equation \eqref{cond-expvalueinit} is the expected value at time $(t+1)^-$ of the investment of the capital funds, and the right hand side is equal to the risk-free return plus the cost of capital rate on the capital funds $SCR_t$ invested at time $t$. We find in the following that the acceptability condition determines the reference portfolio $RP_t$ or allows to derive upper bounds.

\section{Framework for the valuation of insurance liabilities}\label{sec-frame}

At a conceptual level, the proposed framework for market-consistent valuation of an insurance liability $\mathcal{L}$ is based on three ideas:
\begin{enumerate}
\item Dynamic replication of the liability cash-flows by assets given by financial instruments with reliable market prices.
\item Covering the remaining non-replicable part of the cash-flows by capital funds provided by an investor.
\item "Limited liability", i.e. the liability cash-flows in general do not need to be provided for every state of the world.
\end{enumerate}
The first idea is analogous to no-arbitrage or risk-neutral pricing of financial instruments in complete markets. The second idea accounts for the fact that insurance liabilities, in particular, can usually not be perfectly replicated by instruments with reliable market prices, and relates to the requirements by the regulatory authorities, for instance in Solvency II, that companies need to hold a required amount of capital. The third idea relates to the fact that the required regulatory capital typically only needs to be large enough to ensure that the insurance obligations can be satisfied with high probability. In Solvency II, for instance, this is expressed by the 99.5\% Value-at-Risk over a one-year time period.

Valuing the liability $\mathcal{L}$ then means finding a replication procedure, which at a point in time $t$ consists of a portfolio of assets composed of a reference portfolio $RP_t$ and capital funds $C_t$. In a static replication procedure, the portfolio $RP_t$ is held over the lifetime of the liability $\mathcal{L}$. In a discrete dynamic replication procedure, $RP_t$ is dynamically adjusted, in our case (at least) over successive one-year time periods, leading to a sequence of reference portfolios $RP_t$, $RP_{t+1}$... The capital investment $C_t$ for year $t$ is tied from time $t$ to time $t+1$ and is used to cover cash-flow mismatches between $\mathcal{L}$ and $RP_t$ in year $t$ and to convert the assets at time $t+1$ to the next reference portfolio $RP_{t+1}$. At time $t+1$, new capital funds potentially need to be raised for covering the next time period.

The capital investment is assumed to have the following two properties:
\begin{itemize}
\item As an obligation, the capital investment has lowest seniority (i.e. the capital funds are used for covering all other obligations).
\item The capital investment comes with limited liability (i.e. its value is never negative).
\end{itemize}
The crucial assumption about the capital investment is the \emph{acceptability condition}: Under which conditions is the stochastic return from the capital investment acceptable to the capital provider? The acceptability condition specifies the risk-return preferences of the capital investor and is the one input to the framework in addition to the current and future market prices of the financial instruments available for replication.

If an acceptability condition is specified and the reference portfolio $RP_t$ is set up such that the capital investment $C_t$ fulfils the acceptability condition, then the value of $\mathcal{L}$ at time $t$ is defined as in \eqref{def-value} to be the market price of the reference portfolio $RP_t$,
\begin{equation}\label{def-initialvalueofl}
V_{t}(\mathcal{L}):=V_{t}(RP_t).
\end{equation}
The implicit assumption is that required capital funds can always be raised if an acceptable (stochastic) return can be provided. In general, \eqref{def-initialvalueofl} only holds at the point in time at which the corresponding reference portfolio is set up and not in between.\footnote{Moreover, to specify acceptability of the stochastic future value of the capital investment, we have to specify at which time the capital amount $C_t$ is determined, as this is the date at which acceptability of the return to the capital provider is decided. In the following, we assume that $C_t$ is determined at time $t$ and not before.} A major question which we only partially consider in this paper is the uniqueness of the value defined according to \eqref{def-initialvalueofl}.

In view of the third idea underlying the proposed valuation approach, there is the further complication that we allow for limited liability in the replication procedure by limiting the required capital $C_t$. That is, the liability $\mathcal{L}$ does not need to be replicated for every state of the world.

In a dynamic multi-period replication procedure, limited liability potentially applies both backwards and forward in time. Limited liability applies backwards in time because at any point in time we do not only reflect the defaults in the current time period, but additionally the defaults in any future time period.

Limited liability also applies forward in time, in the sense that, at time $t$, there are states of the world in which default has already occurred at a prior point in time. If the liability $\mathcal{L}$ is considered to be a contract with a specific company, this means that, in such a state, the company has defaulted on its obligations prior to $t$, and so the obligations towards future cash-flows cannot be fulfilled anymore to the extent required. We use a different approach, which appears reasonable from the perspective of an insurance regulator, and consider the value at time $t$ of the liability "as such", characterized by future cash-flows and future limited liability, disregarding the replication history prior to time $t$.

\section{Valuation under the framework}\label{sec-oneyear}

The valuation of the liability $\mathcal{L}$ according to \eqref{def-initialvalueofl} is achieved by calculating recursively backwards in time, starting at the end of the lifetime of the liability. Let $T$ denote the final year of the lifetime of $\mathcal{L}$, i.e. $T$ is the smallest whole number such that $X_{T+1}, X_{T+2}\ldots =0$. Then,
\begin{equation*}
V_{T+1}(\mathcal{L})=0.
\end{equation*}
In the recursion step, we assume that the value $V_{t+1}(\mathcal{L})$ at time $t+1$ is known and equal to the market price of a reference portfolio $RP_{t+1}$,
\begin{equation*}
V_{t+1}(\mathcal{L})=V_{t+1}(RP_{t+1}).
\end{equation*}
We then have to calculate the value $V_{t}(\mathcal{L})$ at time $t$ as the market price of a suitable reference portfolio $RP_t$. To this end, define the random variable $Y_{t+1}$ to be the sum of the cash-flow $X_t$ in year $t$ and the value $V_{t+1}(\mathcal{L})$ at the end of the year,
\begin{equation}\label{def-ytp1}
Y_{t+1}:=X_t+V_{t+1}(\mathcal{L}).
\end{equation}
In particular, $Y_{T+1}=X_T$.

For the dynamic replication in year $t$, the random variable $Y_{t+1}$ needs to be matched by assets given by a reference portfolio $RP_t$ together with capital funds $C_t\geq 0$ provided for one year by a capital investor. The capital funds $C_t$ are assumed to be invested at time $t$ in a risk-free one-year zero-coupon bond. We allow for the fact that the replication cannot always be continued past time $t+1$.

To formalize these assumptions, given a reference portfolio $RP_t$ and capital funds $C_t$, the set $A_t$ is defined to be the set of states in which the cash-flow $X_t$ can be provided and the replication can be continued past time $t+1$ by converting the assets available at time $t+1$,
\begin{equation*}
V_{t+1}(RP_t)+\text{tv}_{(t\to t+1)}(C_t)-X_t,
\end{equation*}
to the new reference portfolio $RP_{t+1}$. The set $A_t$ and its probability $\gamma_t$ are thus given by
\begin{eqnarray}
A_t &  := &  \left\{ Y_{t+1}\le \text{tv}_{(t\to t+1)}(C_t)+V_{t+1}(RP_{t})\right\},\label{def-at}\\
\gamma_t &  := &  \mathbb{P}\{A_t\mid\mathcal{F}_t\}=\mathbb{E}\{1_{A_t}\mid\mathcal{F}_t\}.\nonumber
\end{eqnarray}
In view of the characteristics of the capital investment outlined in Section~\ref{sec-frame}, the value of the capital investment at time $t+1$ is given by the maximum of zero and the value of the assets left after all other obligations have been considered, so the value $\tilde{C}_t$ to the capital provider at time $t+1$ of the capital investment $C_t$ can be written as
\begin{equation}\label{def-tildect}
\tilde{C}_t:=1_{A_t}\cdot \left(\text{tv}_{(t\to t+1)}(C_t)+V_{t+1}(RP_t)- Y_{t+1}\right).
\end{equation}
The acceptability condition is specified in the remainder of the paper as prescribed under Solvency II and, in particular, expressed in terms of the expected value of the capital investment. Corresponding to \eqref{cond-expvalueinit}, the acceptability condition is defined to be the condition that the expected excess return over risk-free of the capital investment be equal to a given $\mathcal{F}_t$-measurable "dividend" $D_t\geq 0$,
\begin{equation}\label{def-limliab}
\mathbb{E}\{ \tilde{C}_t \mid\mathcal{F}_t\}-\text{tv}_{(t\to t+1)}(C_t)=D_t.
\end{equation}
The acceptability condition \eqref{def-limliab} translates into an equivalent condition on the reference portfolio: if we insert the expression \eqref{def-tildect} for $\tilde{C}_t$ into \eqref{def-limliab}, we get the condition on the reference portfolio $RP_t$ that
\begin{equation}\label{cond-expvaluewidetilderpt}
\mathbb{E}\{ 1_{A_t}\cdot V_{t+1}(RP_t)\mid\mathcal{F}_t\}=\mathbb{E}\{ 1_{A_t}\cdot Y_{t+1}\mid\mathcal{F}_t\}+(1-\gamma_t)\cdot\text{tv}_{(t\to t+1)}(C_t)+D_t.
\end{equation}
Note that condition \eqref{cond-expvaluewidetilderpt} is complicated in the sense that it depends on $RP_t$, $C_t$, $D_t$, and $A_t$, all of which are in general interlinked with each other.

The value $V_t(\mathcal{L})$ can then be defined in the following way: Given $Y_{t+1}$ defined in \eqref{def-ytp1}, a reference portfolio $RP_t$, capital funds $C_t$, the set $A_t$ and a dividend $D_t$ such that the acceptability condition \eqref{def-limliab} or equivalently \eqref{cond-expvaluewidetilderpt} is satisfied, the value $V_t(\mathcal{L})$ of the insurance liability $\mathcal{L}$ at time $t$ is defined to be the market price of the reference portfolio,
\begin{equation}\label{def-valuebyrpt}
V_t(\mathcal{L}):=V_t(RP_t).
\end{equation}
This immediately entails two questions: Does there always exist a solution to \eqref{def-limliab}, i.e. can a value always be defined by \eqref{def-valuebyrpt}? If so, is such a solution unique, i.e. is the value defined by \eqref{def-valuebyrpt} unique? We provide partial answers to these questions below, but we do not investigate the general question of the uniqueness of the value. In particular, note that the value defined by \eqref{def-valuebyrpt} in general depends on the set $A_t$.

In this respect, we stress that we are not suggesting a "new" definition of the market-consistent value; all we claim to have done so far is provide a precise and more general formulation of the valuation approach for insurance liabilities from Solvency II. The Solvency II approach follows from the general framework by the following three assumptions:
\begin{enumerate}
 \item The capital $C_t$ is given in terms of the reference portfolio $RP_t$ by $SCR_t$ defined in \eqref{def-scrt0} (compare with \eqref{expr-scrt0}), i.e. for a translation-invariant risk measure $\rho$,
    \begin{equation}\label{def-ctinlemma1}
        C_t:=\text{pv}_{(t+1\to t)}\left(\rho\{ V_{t+1}(RP_t)-Y_{t+1}\mid\mathcal{F}_t\}\right).
    \end{equation}
 \item $\rho$ is given as in \eqref{def-rhosii} by the Value-at-Risk $VaR$ at the $99.5\%$ percentile.
 \item $D_t$ is defined as a constant cost of capital rate $\eta >0$ times the capital, i.e.
    \begin{equation}\label{def-dividendtbycoc}
    D_t:=\eta\cdot C_t.
    \end{equation}
\end{enumerate}
In addition, the current prescriptions from EIOPC/SEG/IM13/2010 \cite{im13} suggest that the reference portfolio $RP_t$ should be selected to minimize the capital $C_t$. This can be thought of as a requirement to ensure the uniqueness of the value. However, with the Solvency II selection of $\rho$ as the $99.5\%$ VaR, the capital $C_t$ according to \eqref{def-ctinlemma1}, the set $A_t$ from \eqref{def-at}, and the acceptability condition \eqref{cond-expvaluewidetilderpt} are not affected by values of the difference $Y_{t+1}-V_{t+1}(RP_t)$ beyond their $99.5\%$-quantile, which suggests there might not be uniqueness even if capital is minimized. Of course, an immediate way to ensure uniqueness would be to define the value as the minimum or infimum of the market prices at time $t$ of all reference portfolios $RP_t$ satisfying the acceptability condition \eqref{def-limliab} for the same $C_t$ and $A_t$.

\vspace{10pt}

In the following, we first investigate the existence of solutions to condition \eqref{def-limliab} under two different approaches. Next, we derive in Lemma~\ref{lemma-valueupperbound} an upper bound on solutions of \eqref{def-limliab} and thus on the value defined by \eqref{def-valuebyrpt}. Finally, we show in Theorem~\ref{prop-valueforriskfree} that a unique solution exists and can be explicitly calculated if we assume that the reference market consists only of risk-free zero-coupon bonds and that the capital $C_t$ is defined according to \eqref{def-ctinlemma1}.

For the following proposition, we define an eligible dividend as follows:
\begin{definition}
An $\mathcal{F}_t$-measurable dividend $D_t\geq 0$ from \eqref{def-limliab} is called an \emph{eligible dividend} if, given $\mathcal{F}_t$, $D_t$ is a continuous and monotonously increasing function of $C_t$ with $D_t=0$ for $C_t=0$.
\end{definition}
Clearly, the dividend $D_t$ defined by \eqref{def-dividendtbycoc} is eligible.

We now show that solutions to the acceptability condition \eqref{def-limliab} exist given a suitable form of the set $A_t$ or the capital funds $C_t$.
\begin{prop}\label{lemma-existenceofsolution}
Let $D_t$ be an eligible dividend, and let $Y_{t+1}$ from \eqref{def-ytp1} be given.
\begin{itemize}
\item[(a)] Let $RP_t^{(0)}$ be a reference portfolio and define the set $A_t^{(0)}$ by
    \begin{equation*}
    A_t^{(0)}  :=  \{ Y_{t+1}\le V_{t+1}(RP_{t}^{(0)})\}.
    \end{equation*}
    Then, there exists a capital amount $C_t\geq 0$ and a reference portfolio $RP_t$ such that the corresponding set $A_t$ defined by \eqref{def-at} is equal to $A_t^{(0)}$ and the acceptability condition \eqref{def-limliab} is satisfied.
\item[(b)] Let the capital $C_t$ be given as a function of a reference portfolio $RP_t$ by \eqref{def-ctinlemma1} for a translation-invariant risk measure $\rho$.
     Given a reference portfolio $RP_t^{(0)}$, let the corresponding capital $C_t^{(0)}$ be given by \eqref{def-ctinlemma1} and $A_t^{(0)}$ by \eqref{def-at}.

    Then, there exists a reference portfolio $RP_t$ with the corresponding capital $C_t\geq 0$ given by \eqref{def-ctinlemma1} and the set $A_t$ given by \eqref{def-at} such that $A_t=A_t^{(0)}$ and the acceptability condition \eqref{def-limliab} is satisfied.
\end{itemize}
\end{prop}
\begin{proof}
To prove \emph{(a)}, we split up the portfolio $RP_t^{(0)}$ into a reference portfolio $RP_t$ and capital funds $C_t\geq 0$ by removing from $RP_t^{(0)}$ a one-year risk-free zero-coupon bond with value $C_t\geq 0$ to be determined (or going short in the bond). Then,
\begin{equation}\label{equ-vtp1rpt0}
V_{t+1}(RP_{t}^{(0)})=V_{t+1}(RP_{t})+\text{tv}_{(t\to t+1)}(C_t),
\end{equation}
and the acceptability condition \eqref{cond-expvaluewidetilderpt} for $RP_t$ and $C_t$ can be written as the condition on $\text{tv}_{(t\to t+1)}(C_t)+D_t$ that
\begin{equation*}
\text{tv}_{(t\to t+1)}(C_t)+D_t= \mathbb{E}\{ 1_{A_t}\cdot (V_{t+1}(RP_t^{(0)})-Y_{t+1})\mid\mathcal{F}_t\}\geq 0,
\end{equation*}
because \eqref{equ-vtp1rpt0} ensures that the set $A_t$, if defined by \eqref{def-at} for $RP_t$ and $C_t$, is equal to $A_t^{(0)}$, and the far right inequality above holds by definition of $A_t^{(0)}$. If equality holds in the far right inequality, then the acceptability condition is satisfied for $C_t:=0$ and $RP_t:=RP_t^{(0)}$. If not, then the eligibility of the dividend ensures that we find $C_t>0$ such that the acceptability condition holds.

To prove \emph{(b)}, we use a similar approach as for (a), removing a one-year risk-free zero-coupon bond from $RP_t^{(0)}$ to get a new portfolio $RP_t$. The corresponding capital $C_t$ given by \eqref{def-ctinlemma1} then increases by the corresponding amount because of translation-invariance of the risk measure $\rho$, so
\begin{equation*}
V_{t+1}(RP_{t})+\text{tv}_{(t\to t+1)}(C_t)=V_{t+1}(RP_{t}^{(0)})+\text{tv}_{(t\to t+1)}(C_t^{(0)}),
\end{equation*}
hence the set $A_t$ defined by \eqref{def-at} for $RP_t$ and $C_t$ is equal to $A_t^{(0)}$, and
\begin{eqnarray*}
& & \mathbb{E}\{ 1_{A_t}\cdot (V_{t+1}(RP_t^{(0)})+\text{tv}_{(t\to t+1)}(C_t^{(0)})-Y_{t+1})\mid\mathcal{F}_t\}\\
& = & \gamma_t\cdot\text{tv}_{(t\to t+1)}(C_t)+\mathbb{E}\{ 1_{A_t}\cdot (V_{t+1}(RP_t)-Y_{t+1})\mid\mathcal{F}_t\},
\end{eqnarray*}
so using the acceptability condition \eqref{cond-expvaluewidetilderpt}, we get
\begin{equation*}
\text{tv}_{(t\to t+1)}(C_t)+D_t = \mathbb{E}\{ 1_{A_t}\cdot (V_{t+1}(RP_t^{(0)})+\text{tv}_{(t\to t+1)}(C_t^{(0)})-Y_{t+1})\mid\mathcal{F}_t\}\geq 0
\end{equation*}
by definition of the set $A_t^{(0)}$. The argument then proceeds similarly to (a).
\end{proof}
Next, we provide an upper bound on any solution of \eqref{def-limliab}.
\begin{lemma}\label{lemma-valueupperbound}
Any solution $RP_t$ to the acceptability condition \eqref{def-limliab} and equivalently \eqref{cond-expvaluewidetilderpt} satisfies
\begin{equation*}
\mathbb{E}\{ V_{t+1}(RP_t)\mid\mathcal{F}_t\}\leq \mathbb{E}\{ Y_{t+1}\mid\mathcal{F}_t\} + D_t.
\end{equation*}
\end{lemma}
\begin{proof}
By the definition \eqref{def-at} of $A_t$, we have on the complement $A_t^{c}$ of $A_t$,
\begin{equation*}
1_{A_t^{c}}\cdot V_{t+1}(RP_t)<  1_{A_t^{c}}\cdot Y_{t+1}-1_{A_t^{c}}\cdot \text{tv}_{(t\to t+1)}(C_t).
\end{equation*}
Taking the expected value conditional on $\mathcal{F}_t$ of this expression and adding the result to \eqref{cond-expvaluewidetilderpt}, we get the claimed inequality.
\end{proof}
If we assume in addition that the expected return on $RP_t$ over year $t$ is not less than the risk-free return, then we get from Lemma~\ref{lemma-valueupperbound} the recursive upper bound on $V_t(\mathcal{L})$:
\begin{equation}\label{expr-upperboundtrecurs}
V_t(\mathcal{L})= V_{t}(RP_t)\leq \text{pv}_{(t+1\to t)}\left(\mathbb{E}\{ X_{t}\mid\mathcal{F}_t\} + \mathbb{E}\{ V_{t+1}(\mathcal{L})\mid\mathcal{F}_t\}+D_t\right).
\end{equation}
Under suitable assumptions, we can derive a closed formula upper bound from this recursive inequality.
\begin{prop}\label{prop-upperboundonvalue}
Assume that \eqref{assumpt-forwardprice2} holds and that, for any $t\le s-1$,
\begin{equation}\label{assumpt-xsdsrtindep}
\mathbb{E}\left\{ X_s+D_s\mid \mathcal{F}_{t+1}\right\}\mbox{ and }R_{t+1}^{(s-t)}\mbox{ are independent conditional on $\mathcal{F}_t$}.
\end{equation}
Assume that, for any year $t$, the expected return on any reference portfolio $RP_t$ is larger than or equal to the risk-free return. Further assume that, for any $t$, the set $A_t$ is given by \eqref{def-at}. Then, the value $V_t(\mathcal{L})$ at time $t$ of the liability $\mathcal{L}$ is bounded above by
\begin{equation}\label{expr-upperboundt}
V_t(\mathcal{L})\leq \sum_{s=t}^T \text{pv}_{(s+1\to t)}\left(\mathbb{E}\left\{ X_s+D_s\mid \mathcal{F}_t\right\}\right).
\end{equation}
\end{prop}
\begin{proof}
We proceed by induction backwards in time, starting from $t=T$. For $t=T$, the claim is given by \eqref{expr-upperboundtrecurs}, since $V_{T+1}(\mathcal{L})=0$. Now assume that \eqref{expr-upperboundt} holds for $t+1$, i.e., using the notation \eqref{def-annualriskfreerate},
\begin{equation}\label{expr-upperboundtp1}
V_{t+1}(\mathcal{L})\leq \sum_{s=t+1}^T \left(1+R_{t+1}^{(s-t)}\right)^{-(s-t)}\cdot\mathbb{E}\left\{ X_s+D_s\mid \mathcal{F}_{t+1}\right\}.
\end{equation}
The recursive upper bound \eqref{expr-upperboundtrecurs} for $t$ can be written
\begin{equation*}
V_t(\mathcal{L})\leq (1+R_t^{(1)})^{-1}\cdot \mathbb{E}\{X_{t}+D_t\mid\mathcal{F}_t\} +(1+R_t^{(1)})^{-1}\cdot \mathbb{E}\{V_{t+1}(\mathcal{L})\mid\mathcal{F}_t\}.
\end{equation*}
Inserting \eqref{expr-upperboundtp1} into this inequality, using \eqref{assumpt-xsdsrtindep} and applying first \eqref{assumpt-forwardprice2} and then \eqref{eq-noarbitrageforward} proves \eqref{expr-upperboundt}.
\end{proof}

If we assume that the reference market consists only of risk-free zero-coupon bonds, then the acceptability condition \eqref{cond-expvaluewidetilderpt} on the reference portfolio $RP_t$ explicitly determines the reference portfolio, given $Y_{t+1}$, $C_t$ and $D_t$, and provided that $\gamma_t >0$. In fact, $V_{t+1}(RP_t)= \text{tv}_{(t\to t+1)}(V_{t}(RP_t))$ is then $\mathcal{F}_t$-measurable, so it can be taken out of the expectation in \eqref{cond-expvaluewidetilderpt}, and we get
\begin{equation}\label{expr-vtrptriskfree}
\gamma_t\cdot \text{tv}_{(t\to t+1)}(V_{t}(RP_t))=\mathbb{E}\{ 1_{A_t}\cdot Y_{t+1}\}+(1-\gamma_t)\cdot \text{tv}_{(t\to t+1)}(C_t)+D_t.
\end{equation}
This result can be refined for the special case that the capital $C_t$ is defined in line with Solvency II by \eqref{def-ctinlemma1} to derive an explicit recursive expression for the value of $\mathcal{L}$.
\begin{theorem}\label{prop-valueforriskfree}
Assume that the set $A_t$ is given by \eqref{def-at} and the capital $C_t$ by \eqref{def-ctinlemma1}. Further assume that the reference market consists only of the risk-free zero-coupon bonds. Then, the value $V_t(\mathcal{L})$ at time $t$ of the liability $\mathcal{L}$ is uniquely determined by the recursive expression
\begin{equation}\label{expr-v0mathcalyt}
V_t(\mathcal{L})=\text{pv}_{(t+1\to t)}\left(\mathbb{E}\{ 1_{A_t}\cdot Y_{t+1}\mid\mathcal{F}_t\}+(1-\gamma_t)\cdot \rho\{ -Y_{t+1}\mid\mathcal{F}_t\}+D_t\right),
\end{equation}
where $D_t$ is an eligible dividend from \eqref{def-limliab}. The set $A_t$ and the capital $C_t$ can be written as
\begin{eqnarray}
A_t &  = &  \left\{ Y_{t+1}\le \rho\{ -Y_{t+1}\mid\mathcal{F}_t\}\right\}\label{expr-at},\\
C_t &  = &  \text{pv}_{(t+1\to t)}\left(\rho\{ -Y_{t+1}\mid\mathcal{F}_t\}\right)-V_t(\mathcal{L}).\nonumber
\end{eqnarray}
\end{theorem}
\begin{proof}
As the risk measure $\rho$ is translation-invariant, the capital from \eqref{def-ctinlemma1} is given by
\begin{equation}\label{eq-ctbyrho}
\text{tv}_{(t\to t+1)}(C_t)=\rho\{ -Y_{t+1}\mid\mathcal{F}_t\}-\text{tv}_{(t\to t+1)}(V_{t}(\mathcal{L})),
\end{equation}
so the set $A_t$ is given by \eqref{expr-at}. Inserting \eqref{eq-ctbyrho} into the expression \eqref{expr-vtrptriskfree} for $V_t(RP_t)$ then proves \eqref{expr-v0mathcalyt} as, by definition, $V_t(\mathcal{L})=V_t(RP_t)$.
\end{proof}
Note that it is not obvious how to derive a reasonable closed formula expression from the recursive expression \eqref{expr-v0mathcalyt} because of the $\rho$-term.

A more concise expression for \eqref{expr-v0mathcalyt} can be given if we define the random variable $Z_{t+1}$ as the "cut-off" of $Y_{t+1}$,
\begin{equation*}
Z_{t+1}:=1_{A_t}\cdot Y_{t+1}+1_{A_t^c}\cdot \rho\{ -Y_{t+1}\mid\mathcal{F}_t\}.
\end{equation*}
Then the recursion \eqref{expr-v0mathcalyt} can be written
\begin{equation*}
V_t(\mathcal{L})=\text{pv}_{(t+1\to t)}\left(\mathbb{E}\{ Z_{t+1}\mid\mathcal{F}_t\}+D_t\right).
\end{equation*}

\section{The risk margin}\label{sec-rm}

Recall from the introduction \eqref{def-mcvl} the idea of defining the value of an insurance liability $\mathcal{L}$ by the sum of a "best estimate", which we interpret (maybe more generally than in Solvency II) as the market price of a reference portfolio, and a risk margin corresponding to capital costs. The idea is that the risk margin accounts for the non-hedgeable part of the cash-flows of the liability $\mathcal{L}$ to be valued. However, in the preceding part of the paper, a split in best estimate and risk margin was never required. Moreover, the definition of the risk margin is ambiguous in the context of dynamic multi-period replication, because there are two conflicting intentions: on the one hand, the "best estimate" is thought to capture only the contractual cash-flows of the insurance liability and not capital costs. On the other hand, the "best estimate" should capture the hedgeable part over a one-year time period, which in general partially includes also future capital costs.

We show in the following that it is possible to define a risk margin, and to use it to derive an upper bound on the value. However, the risk margin we define depends on certain assumptions, and other definitions of a risk margin would also be possible.

To define the risk margin, the idea is to split the reference portfolio $RP_t$ into a reference portfolio $\tilde{RP}_t$, whose market price is the "best estimate", and a portfolio we call a "dividend portfolio" $DP_t$, such that the dividend portfolio accounts for all capital costs, and its market price corresponds to the risk margin. So $RP_t$ consists of the two portfolios $\tilde{RP}_t$ and $DP_t$ and, since market prices are additive, the value can then be written as
\begin{equation*}
V_t(\mathcal{L})=V_{t}(RP_t)=V_{t}(\tilde{RP}_t)+V_{t}(DP_t).
\end{equation*}
We assume that $DP_t$ consists of a risk-free one-year zero-coupon bond (compare to the formula \eqref{def-riskmargininitial} for the risk margin in Solvency II). Then, $Y_{t+1}$ from \eqref{def-ytp1} can be written as
\begin{equation}\label{eq-ytp1forsplit}
Y_{t+1}  =   X_t+V_{t+1}(\tilde{RP}_{t+1})+V_{t+1}(DP_{t+1}).
\end{equation}
For deriving the split, the reference portfolios $\tilde{RP}_t$ for every year $t$ are determined first; they account for all future cash-flows $(X_s)_{s\geq t}$ of $\mathcal{L}$ and disregard limited liability. That is, they disregard the fact that the replication cannot always be continued. The dividend portfolios $DP_t$ are constructed afterwards.

Define $\overline{RM}_t$ be the "expected risk margin" at time $t$,
\begin{equation*}
\overline{RM}_t:=\sum_{s=t}^{T}\text{pv}_{(s+1\to t)}\left(\mathbb{E}\left\{ D_s\mid \mathcal{F}_t\right\}\right),
\end{equation*}
where $D_s$ for $s\geq t$ are eligible dividends from \eqref{def-limliab}.

We now prove one of our main results, which is that the sum of the market price of the reference portfolio $\tilde{RP}_t$ and the "expected risk margin" $\overline{RM}_t$ is an upper bound for the value of $\mathcal{L}$, see \eqref{expr-upperboundt2} as well as Corollary~\ref{cor-upperboundriskmargin} below.

\begin{theorem}\label{prop-upperboundonvaluedividendpf}
Let the set $A_t$ on which the replication can be continued be given by \eqref{def-at},
\begin{equation*}
A_t  = \left\{ Y_{t+1}\le \text{tv}_{(t\to t+1)}(C_t)+V_{t+1}(RP_{t})\right\}.
\end{equation*}
Assume that the reference portfolio $RP_t$ consists of a dividend portfolio $DP_t$ in the form of a one-year risk-free zero-coupon bond and a reference portfolio $\tilde{RP}_t$, and that the acceptability condition \eqref{def-limliab} is satisfied for a given eligible dividend $D_t\geq 0$.
\begin{itemize}
\item[(a)] Assume that \eqref{assumpt-forwardprice2} holds, that the portfolio $\tilde{RP}_t$ satisfies
    \begin{equation}\label{eq-exporpt}
    \mathbb{E}\{ V_{t+1}(\tilde{RP}_t) \mid\mathcal{F}_t\}\geq \mathbb{E}\{ X_t+V_{t+1}(\tilde{RP}_{t+1}) \mid\mathcal{F}_t\},
    \end{equation}
    and that, for any $t\le s-1$,
    \begin{equation}\label{assumpt-xsdsrtindep2}
    \mathbb{E}\left\{ D_s\mid \mathcal{F}_{t+1}\right\}\mbox{ and }R_{t+1}^{(s-t)}\mbox{ are independent conditional on $\mathcal{F}_t$}.
    \end{equation}

    Then, the value at time $t$ of the liability $\mathcal{L}$ satisfies
    \begin{equation}\label{expr-upperboundt2}
    V_t(\mathcal{L})= V_t(\tilde{RP}_t)+V_t(DP_t)\;\;\mbox{with}\;\;\; V_t(DP_t)\leq \overline{RM}_t.
    \end{equation}
\item[(b)] Let the capital $C_t$ be given by \eqref{def-ctinlemma1} for a translation-invariant risk measure $\rho$,
    \begin{equation*}
    C_t=\text{pv}_{(t+1\to t)}\left(\rho\{ V_{t+1}(RP_t)-Y_{t+1}\mid\mathcal{F}_t\}\right).
    \end{equation*}
    Then $DP_t$ is given by the recursive expression
    \begin{eqnarray}
    & & V_t(DP_t)=\text{pv}_{(t+1\to t)}\left(\mathbb{E}\{ 1_{A_t}\cdot ( Y_{t+1}-V_{t+1}(\tilde{RP}_t))\mid\mathcal{F}_t\}\right)+\label{expr-dptrecursive}\\
    & & + \text{pv}_{(t+1\to t)}\left((1-\gamma_t)\cdot \rho\{ V_{t+1}(\tilde{RP}_t)-Y_{t+1}\mid\mathcal{F}_t\}+D_t\right)\nonumber ,
    \end{eqnarray}
    where
    \begin{eqnarray}
    C_t &  = &  \text{pv}_{(t+1\to t)}\left(\rho\{ V_{t+1}(\tilde{RP}_t)-Y_{t+1}\mid\mathcal{F}_t\}\right)-V_{t}(DP_{t})\label{expr-ctrm},\\
    A_t &  = &  \left\{ Y_{t+1}-V_{t+1}(\tilde{RP}_{t})\le \rho\{ V_{t+1}(\tilde{RP}_{t}) -Y_{t+1}\mid\mathcal{F}_t\}\right\}\nonumber .
    \end{eqnarray}
\end{itemize}
\end{theorem}
\begin{proof}
To prove \emph{(a)}, we need to prove \eqref{expr-upperboundt2}, i.e. $V_t(DP_t)\leq \overline{RM}_t$. To this end, we first derive a recursive expression for $V_t(DP_t)$. In fact, from Lemma~\ref{lemma-valueupperbound}, we have
\begin{equation*}
\mathbb{E}\left\{ V_{t+1}(RP_t)\mid \mathcal{F}_{t}\right\}\leq \mathbb{E}\left\{ Y_{t+1}\mid \mathcal{F}_{t}\right\}+D_t.
\end{equation*}
Inserting $\tilde{RP}_t$, $DP_t$, and $Y_{t+1}$ from \eqref{eq-ytp1forsplit} into this inequality and using \eqref{eq-exporpt}, we get, since $DP_t$ consists of a one-year risk-free zero-coupon bond,
\begin{equation}\label{expr-recboundondptp1}
V_{t+1}(DP_t)=\mathbb{E}\left\{ V_{t+1}(DP_t)\mid \mathcal{F}_{t}\right\}\leq \mathbb{E}\left\{ V_{t+1}(DP_{t+1})\mid \mathcal{F}_{t}\right\}+D_t.
\end{equation}
This implies the recursive upper bound on $V_t(DP_t)$,
\begin{equation*}
V_{t}(DP_t)\leq \text{pv}_{(t+1\to t)}\left(\mathbb{E}\left\{ V_{t+1}(DP_{t+1})\mid \mathcal{F}_{t}\right\}+D_t\right).
\end{equation*}
Arguing as in the proof of Proposition~\ref{prop-upperboundonvalue} and using \eqref{assumpt-xsdsrtindep2}, the upper bound \eqref{expr-upperboundt2} then follows.

To prove \emph{(b)}, the recursive expression \eqref{expr-dptrecursive} follows from the acceptability condition \eqref{cond-expvaluewidetilderpt} by similar arguments as Theorem~\ref{prop-valueforriskfree}, using that $V_{t+1}(DP_t)$ is $\mathcal{F}_t$-measurable and that $\rho$ is translation-invariant.
\end{proof}

\begin{remark}\label{rem-optimalbound}
Note that the upper bound $V_t(DP_t)\leq \overline{RM}_t$ from Theorem~\ref{prop-upperboundonvaluedividendpf} (a) holds regardless of whether equality holds in \eqref{eq-exporpt} or not. This means that the upper bound $\overline{RM}_t$ on the value $V_t(DP_t)$ is not affected by the selection of $\tilde{RP}_t$ subject to \eqref{eq-exporpt}, although $V_t(DP_t)$ is. In order to obtain the most useful upper bound on the value, one should thus select $\tilde{RP}_t$ as that reference portfolio satisfying \eqref{eq-exporpt} which minimizes $V_{t}(\tilde{RP}_t)$.
\end{remark}
\begin{remark}
The upper bound from Theorem~\ref{prop-upperboundonvaluedividendpf} (a) holds in particular if the portfolio $\tilde{RP}_t$ is assumed to consist of a one-year risk-free zero-coupon bond and is defined to be the reference portfolio matching the expected values of the cash-flows $(X_s)_{s\geq t}$ of the liability $\mathcal{L}$ to be valued by risk-free zero-coupon bonds, i.e. if the value $V_t(\tilde{RP}_t)$ is given by
\begin{equation*}
V_t(\tilde{RP}_t)=\sum_{s=t}^{T}\text{pv}_{(s+1\to t)}\left(\mathbb{E}\left\{ X_s\mid \mathcal{F}_t\right\}\right).
\end{equation*}
This is also the reference portfolio which is "optimal" in the sense of minimizing $V_{t}(\tilde{RP}_t)$ as in Remark~\ref{rem-optimalbound} and for which equality holds in \eqref{eq-exporpt}.
\end{remark}
If we assume that the dividend $D_t$ is given as in Solvency II by a constant cost of capital rate applied to the capital, see \eqref{def-dividendtbycoc}, then we get the following corollary.
\begin{corollary}\label{cor-upperboundriskmargin}
Assume that \eqref{assumpt-forwardprice2} holds. Assume that the set $A_t$ on which the replication can be continued is given by \eqref{def-at},
\begin{equation*}
A_t  = \left\{ Y_{t+1}\le \text{tv}_{(t\to t+1)}(C_t)+V_{t+1}(RP_{t})\right\},
\end{equation*}
with the capital $C_t$ given by \eqref{def-ctinlemma1} for a translation-invariant risk measure $\rho$,
\begin{equation*}
C_t=\text{pv}_{(t+1\to t)}\left(\rho\{ V_{t+1}(RP_t)-Y_{t+1}\mid\mathcal{F}_t\}\right).
\end{equation*}
Assume that the acceptability condition \eqref{def-limliab} is satisfied for the dividend $D_t$ given by \eqref{def-dividendtbycoc}, i.e. $D_t=\eta\cdot C_t$ for some $\eta>0$, and that the reference portfolio $RP_t$ consists of a dividend portfolio $DP_t$ in the form of a one-year risk-free zero-coupon bond and a reference portfolio $\tilde{RP}_t$ satisfying \eqref{eq-exporpt},
\begin{equation*}
\mathbb{E}\{ V_{t+1}(\tilde{RP}_t) \mid\mathcal{F}_t\}\geq \mathbb{E}\{ X_t+V_{t+1}(\tilde{RP}_{t+1}) \mid\mathcal{F}_t\}.
\end{equation*}
Define
\begin{equation*}
\tilde{Y}_{t+1}:=X_t+V_{t+1}(\tilde{RP}_{t+1}),\;\;\;\; \Delta DP_{t+1}:=V_{t+1}(DP_{t+1})-\mathbb{E}\{V_{t+1}(DP_{t+1}) \mid\mathcal{F}_t\},
\end{equation*}
and assume that, for any $t\le s-1$,
\begin{eqnarray}
\mathbb{E}\left\{\rho\{ V_{s+1}(\tilde{RP}_{s}) -\tilde{Y}_{s+1}-\Delta DP_{s+1}\mid\mathcal{F}_s\}\mid \mathcal{F}_{t+1}\right\}\label{assumpt-xsdsrtindep3}\\
\mbox{ and }R_{t+1}^{(s-t)}\mbox{ are independent conditional on $\mathcal{F}_t$.}\nonumber
\end{eqnarray}
Then we have the upper bound
\begin{equation*}
V_t(\mathcal{L})\leq V_t(\tilde{RP}_t)+\tilde{RM}_t
\end{equation*}
with the "adjusted expected risk margin"
\begin{equation}\label{def-tildermt}
\tilde{RM}_t:=\frac{\eta}{1+\eta}\cdot \sum_{s=t}^{T}\text{pv}_{(s+1\to t)}\left(\mathbb{E}\left\{ \rho\{ V_{s+1}(\tilde{RP}_{s}) -\tilde{Y}_{s+1}-\Delta DP_{s+1}\mid\mathcal{F}_s\}\mid \mathcal{F}_t \right\}\right).
\end{equation}
\end{corollary}

\begin{proof}
We show that $V_{t}(DP_t)\leq \tilde{RM}_t$ for any $t$. To this end, we note that the assumptions of Theorem~\ref{prop-upperboundonvaluedividendpf} are satisfied, and so we can insert $D_t=\eta\cdot C_t$ and the expression \eqref{expr-ctrm} for the capital $C_t$ into \eqref{expr-recboundondptp1} to get
\begin{equation*}
(1+R_t^{(1)}+\eta)\cdot V_{t+1}(DP_t)\leq (1+R_t^{(1)})\cdot\mathbb{E}\left\{ V_{t+1}(DP_{t+1})\mid \mathcal{F}_{t}\right\}+\eta \cdot \rho\{ V_{t+1}(\tilde{RP}_{t}) -Y_{t+1}\mid\mathcal{F}_t\}.
\end{equation*}
Using that $\rho$ is translation-invariant and in view of the definitions of $\tilde{Y}_{t+1}$ and $\Delta DP_{t+1}$, and with the inequality $1+\eta\leq 1+R_t^{(1)}+\eta$, we can write this as the recursive upper bound
\begin{equation*}
V_{t+1}(DP_t)\leq \mathbb{E}\left\{ V_{t+1}(DP_{t+1})\mid \mathcal{F}_{t}\right\}+\frac{\eta}{1+\eta} \cdot \rho\{ V_{t+1}(\tilde{RP}_{t}) -\tilde{Y}_{t+1}-\Delta DP_{t+1}\mid\mathcal{F}_t\}.
\end{equation*}
Arguing as in the proof of Proposition~\ref{prop-upperboundonvalue} and using \eqref{assumpt-xsdsrtindep3}, the upper bound then follows.
\end{proof}
In practice, it is often assumed as a simplification that "the risk margin does not contribute to the one-year volatility", i.e.
\begin{equation*}
V_{t+1}(DP_{t})\approx V_{t+1}(DP_{t+1})+D_t.
\end{equation*}
This assumption is implicit in the Swiss Solvency Test (SST) and is often assumed in Solvency II when the Solvency Capital Requirement needs to be calculated. In our case, we can formulate the corresponding condition as (for some small $\varepsilon>0$)
\begin{equation*}
\Delta DP_{t+1}=V_{t+1}(DP_{t+1})-\mathbb{E}\{V_{t+1}(DP_{t+1}) \mid\mathcal{F}_t\}\leq \varepsilon\cdot \rho\{ V_{t+1}(\tilde{RP}_{t}) -\tilde{Y}_{t+1}\mid\mathcal{F}_t\}.
\end{equation*}
I.e. the possible increase in the estimate of $V_{t+1}(DP_{t+1})$ from time $t$ to $t+1$ is small compared to the $\rho$-term. If we assume this holds for any $t$ and, in addition, that the risk measure $\rho$ is monotone, then we get
\begin{eqnarray*}
& & \rho\{ V_{t+1}(\tilde{RP}_{t}) -\tilde{Y}_{t+1}-\Delta DP_{t+1}\mid\mathcal{F}_t\}  \\
& \leq & \rho\{ V_{t+1}(\tilde{RP}_{t}) -\tilde{Y}_{t+1}-\varepsilon\cdot \rho\{ V_{t+1}(\tilde{RP}_{t}) -\tilde{Y}_{t+1}\mid\mathcal{F}_t\}\mid\mathcal{F}_t\}\\
& = & (1+ \varepsilon)\cdot \rho\{ V_{t+1}(\tilde{RP}_{t}) -\tilde{Y}_{t+1}\mid\mathcal{F}_t\}
\end{eqnarray*}
using that $\rho$ is translation-invariant. We then get for the "adjusted expected risk margin" $\tilde{RM}_t$ from \eqref{def-tildermt} the upper bound
\begin{equation*}
\tilde{RM}_t\leq \frac{\eta\cdot (1+\varepsilon)}{1+\eta}\cdot \sum_{s=t}^{T}\text{pv}_{(s+1\to t)}\left(\mathbb{E}\left\{ \rho\{ V_{s+1}(\tilde{RP}_{s}) -X_s-V_{s+1}(\tilde{RP}_{s+1})\mid\mathcal{F}_s\}\mid \mathcal{F}_t \right\}\right).
\end{equation*}
This expression for the risk margin appears to be the one implicitly used in most actual calculations in the context of Solvency II (and the SST), but without the $\eta$-term in the denominator of the fraction above and with $\varepsilon$ set equal to zero.

Note that the capital for any year is calculated above in terms of the one-year change of the "best estimate", and not in terms of the one-year change of the difference between the value of the assets covering the "best estimate" and the value of the liability.

\section{An example for the calculation of the value}\label{sec-ex}

We now explicitly calculate the value for a simple example with two liabilities $\mathcal{L}_1$ and $\mathcal{L}_2$. The example illustrates the upper bound from Theorem~\ref{prop-upperboundonvaluedividendpf} and shows that the "ordering" of the value of two liabilities can change: there are instances in which the value of one liability is larger than the other but the inequality is reversed between the upper bounds. We note that the example we present is not realistic as we assume that the risk-free interest rate is zero and that successive cash-flows are independent. However, it may be surprising that the change of the "ordering" occurs even under these assumptions.

We assume that only risk-free zero-coupon bonds are eligible for the replication and that the risk-free interest rate is zero. These two conditions imply that actually only cash is available. We further assume that the capital $C_t$ is given by \eqref{def-ctinlemma1} with the risk measure $\rho$ as in \eqref{def-rhosii} given by the Value-at-Risk at a confidence level $0<\alpha <1$, and that the dividend $D_t$ is given by \eqref{def-dividendtbycoc} with $\eta >0$. Under these assumptions, Theorem~\ref{prop-valueforriskfree} implies that the value $V_t(\mathcal{L})$ of a liability $\mathcal{L}$ is given by
\begin{equation}\label{expr-valueex}
V_t(\mathcal{L})=\frac{1}{1+\eta}\left(\mathbb{E}\{ 1_{A_t}\cdot Y_{t+1}\mid\mathcal{F}_t\}+(1+\eta-\gamma_t)\cdot \rho\{ -Y_{t+1}\mid\mathcal{F}_t\}\right).
\end{equation}
The upper bound on the value according to Theorem~\ref{prop-upperboundonvaluedividendpf} becomes
\begin{eqnarray}
V_t^{u}(\mathcal{L}) & := & \sum_{s=t}^{T} \mathbb{E}\{ X_s\mid\mathcal{F}_t\}+\eta\cdot\sum_{s=t}^{T} \mathbb{E}\{ C_s\mid\mathcal{F}_t\}\label{def-upperboundex}\\
\mbox{with }C_s & = & \rho\{ -Y_{s+1}\mid\mathcal{F}_s\}-V_s(\mathcal{L})\label{expr-capitalex} ,
\end{eqnarray}
where $X_t$ denotes the claims payment of the insurance liability $\mathcal{L}$ in year $t$. $X_t$ is $\mathcal{F}_{t+1}$-measurable, and we assume in addition that its distribution conditional on $\mathcal{F}_s$ for any $0\leq s<t+1$ is independent of $s$, i.e. no information about $X_t$ is revealed before time $t+1$. In particular, the $X_t$ are independent.

We consider the two years $t=0$ and $t=1$ and assume $X_t=0$ for $t\geq 2$. We suppress conditioning on $\mathcal{F}_0$ in the notation, and then get from Theorem~\ref{prop-valueforriskfree} and the above assumption on $X_1$ that
\begin{eqnarray*}
\mathbb{E}\{ 1_{A_1}\cdot Y_{2}\mid\mathcal{F}_1\} & = & \mathbb{E}\{ 1_{\{X_1\leq \rho\{-X_1\mid\mathcal{F}_1\}\}}\cdot X_{1}\mid\mathcal{F}_1\}=\mathbb{E}\{ 1_{\{X_1\leq \rho\{-X_1\}\}}\cdot X_{1}\}\\
\rho\{ -Y_{2}\mid\mathcal{F}_1\} & = & \rho\{ -X_{1}\}.
\end{eqnarray*}
Hence the value at time $t=1$ from \eqref{expr-valueex} is
\begin{equation*}
V_1(\mathcal{L})=\frac{1}{1+\eta}\left(\mathbb{E}\{ 1_{\{X_1\leq \rho\{-X_1\}\}}\cdot X_{1}\}+(1+\eta-\gamma_1)\cdot \rho\{ -X_{1}\}\right).
\end{equation*}
Since $V_1(\mathcal{L})$ thus is $\mathcal{F}_0$-measurable, and $\rho$ is translation-invariant, we get from this and \eqref{expr-valueex} for the value at time $t=0$ that
\begin{eqnarray}
V_0(\mathcal{L}) & = & \frac{1}{1+\eta}\left(\mathbb{E}\{ 1_{\{X_0\leq \rho\{-X_0\}\}}\cdot X_{0}\}+(1+\eta-\gamma_0)\cdot \rho\{ -X_{0}\}+(1+\eta)\cdot V_1(\mathcal{L})\right)\nonumber\\
& = & \frac{1}{1+\eta}\left(\mathbb{E}\{ 1_{\{X_0\leq \rho\{-X_0\}\}}\cdot X_{0}\}+\mathbb{E}\{ 1_{\{X_1\leq \rho\{-X_1\}\}}\cdot X_{1}\}\right)\nonumber\\
& + & \frac{1}{1+\eta}\left((1+\eta-\gamma_0)\cdot \rho\{ -X_{0}\}+(1+\eta-\gamma_1)\cdot \rho\{ -X_{1}\}\right).\nonumber
\end{eqnarray}
With $\rho$ given by the Value-at-Risk at confidence level $0<\alpha <1$, so $\gamma_t=\alpha$, and for a normally distributed random variable $Z$ with mean $\mu$ and standard deviation $\sigma$, we have
\begin{eqnarray*}
\rho \{-Z\} & = & VaR_{\alpha}(Z)=\mu+\sigma\cdot q_{\alpha}\\
\mathbb{E}\{ 1_{\{Z\leq \rho(-Z)\}}\cdot Z\} & = & \mathbb{E}\{ Z\}-(1-\alpha)\cdot \mathbb{E}\{ Z\mid Z> \rho(-Z)\}=\\
& = & \mu-(1-\alpha)\cdot \left(\mu+\frac{\phi (q_{\alpha})}{1-\alpha}\cdot\sigma\right)= \alpha\cdot \mu -\phi (q_{\alpha})\cdot\sigma ,
\end{eqnarray*}
where $\phi$ denotes the probability density function and $q_{\alpha}$ the ${\alpha}$-quantile of the standard normal distribution.

Assume that $X_t$ for $t=0,1$ are normally distributed with mean $\mu_t$ and standard deviation $\sigma_t$. Then the above expressions imply that the values at times $t=0,1$ become
\begin{eqnarray}
V_1(\mathcal{L}) & = & \mu_1+ \frac{\sigma_1}{1+\eta}\cdot ((1+\eta-\alpha)\cdot q_{\alpha} -\phi (q_{\alpha})),\nonumber\\
V_0(\mathcal{L}) & = & \mu_0+\mu_1+\frac{\sigma_0+\sigma_1}{1+\eta}\cdot ((1+\eta-\alpha)\cdot q_{\alpha} -\phi (q_{\alpha})).\label{expr-valueexcalc}
\end{eqnarray}
Inserting this into the expression for the capital amounts $C_0$ and $C_1$ from \eqref{expr-capitalex}, we get
\begin{eqnarray*}
C_0+C_1 & = & \rho \{-X_0-V_1(\mathcal{L})\}-V_0(\mathcal{L})+\rho \{-X_1\}-V_1(\mathcal{L})=\rho \{-X_0\}+\rho \{-X_1\}-V_0(\mathcal{L})\\
& = & \frac{\sigma_0+\sigma_1}{1+\eta}\cdot (\alpha\cdot q_{\alpha} +\phi (q_{\alpha})),
\end{eqnarray*}
hence the upper bound from \eqref{def-upperboundex} on the value at $t=0$ becomes
\begin{equation}\label{expr-upperboundex}
V_0^{u}(\mathcal{L}) = \mu_0+\mu_1+\frac{\eta\cdot (\sigma_0+\sigma_1)}{1+\eta}\cdot (\alpha\cdot q_{\alpha} +\phi (q_{\alpha})).
\end{equation}
Now consider two different liabilities: for the first liability $\mathcal{L}^{(1)}$, we assume as above that $X_t$ for $t=0,1$ is normally distributed with mean $\mu_t>0$ and standard deviation $\sigma_t>0$. So the value $V_0(\mathcal{L}^{(1)})$ and the upper bound $V_0^{u}(\mathcal{L}^{(1)})$ are given as in \eqref{expr-valueexcalc} and \eqref{expr-upperboundex}, respectively. The second liability $\mathcal{L}^{(2)}$ we define by $X_1=0$ and $X_0$ normally distributed with mean $\mu_0+\mu_1$ and standard deviation $(\sigma_0^2+\sigma_1^2)^{\frac{1}{2}}$. Its value $V_0(\mathcal{L}^{(2)})$ and upper bound $V_0^{u}(\mathcal{L}^{(2)})$ are then given as in \eqref{expr-valueexcalc} and \eqref{expr-upperboundex}, respectively, but with $\sigma_0+\sigma_1$ replaced by $(\sigma_0^2+\sigma_1^2)^{\frac{1}{2}}$. We get:
\begin{prop}\label{prop-example1}
Let $\mathcal{L}^{(1)}$ and $\mathcal{L}^{(2)}$ be defined as above.
\begin{itemize}
\item[(a)] The value and upper bounds of the value at time $t=0$ satisfy the inequalities
\begin{equation*}
V_0(\mathcal{L}^{(i)})< V_0^{u}(\mathcal{L}^{(i)})\;\mbox{ for }i=1,2,\;\;\; V_0^{u}(\mathcal{L}^{(2)})< V_0^{u}(\mathcal{L}^{(1)}).
\end{equation*}
\item[(b)] For any $\eta >0$ sufficiently small,
\begin{equation*}
V_0(\mathcal{L}^{(1)})< V_0(\mathcal{L}^{(2)}).
\end{equation*}
\end{itemize}
\end{prop}
\begin{proof}
To prove \emph{(a)}, the first statement (as a non-strict inequality) follows from Theorem~\ref{prop-upperboundonvaluedividendpf}, but we show it here explicitly by observing that the relevant difference, which has to be shown to be positive, can be written
\begin{equation*}
\eta\cdot (\alpha\cdot q_{\alpha} +\phi (q_{\alpha}))-((1+\eta-\alpha)\cdot q_{\alpha}-\phi (q_{\alpha}))=(1+\eta)\cdot f(\alpha),
\end{equation*}
where $f$ is defined by
\begin{equation*}
f(\alpha):=\phi (q_{\alpha})-(1-\alpha)\cdot q_{\alpha}.
\end{equation*}
$f(\alpha)$ is strictly positive for any $0<\alpha <1$. This follows since it clearly holds for $\alpha>0$ close to $0$, and $f$ is strictly monotonously decreasing as
\begin{equation*}
f^{\prime}(\alpha)=-(1-\alpha)\cdot \phi (q_{\alpha})^{-1}<0\;\mbox{ for }0<\alpha <1
\end{equation*}
which follows from
\begin{equation*}
\frac{d}{d\alpha}q_{\alpha}=\phi (q_{\alpha})^{-1},\;\;\; \frac{d}{dx}\phi (x)=-x\cdot \phi (x)
\end{equation*}
and $f(\alpha)\to 0$ for $\alpha\to 1-$. The second statement holds as $(\sigma_0^2+\sigma_1^2)^{\frac{1}{2}}< \sigma_0+\sigma_1$ and
\begin{equation*}
g(\alpha):=\alpha\cdot q_{\alpha} +\phi (q_{\alpha})> 0\;\mbox{ for }0<\alpha <1
\end{equation*}
which can be shown similarly to the statement on $f$.

\emph{(b)} follows from $(\sigma_0^2+\sigma_1^2)^{\frac{1}{2}}< \sigma_0+\sigma_1$ and the strict positivity of $f$ shown in the proof of (a), which implies that, for fixed $\alpha$, we have for $\eta >0$ sufficiently small,
\begin{equation}\label{equ-etaalphaphi}
(1+\eta-\alpha)\cdot q_{\alpha}-\phi (q_{\alpha}) < 0.
\end{equation}
\end{proof}
The inequality between the upper bounds from Proposition~\ref{prop-example1} (a) is not surprising: the total loss $X_1+X_2$ is the same for the two liabilities, but for the first liability $\mathcal{L}^{(1)}$, the total loss is distributed over two years and so (because no information about $X_1$ is revealed before time $t=2$) does not allow for taking into account diversification between $X_0$ and $X_1$. The inequality between the values from Proposition~\ref{prop-example1} (b) goes into the opposite direction for $\eta$ small. In fact, it follows from \eqref{equ-etaalphaphi} that, the smaller the safety level $\alpha$, i.e. the larger the probability that "limited liability" applies, the larger the cost of capital rate $\eta$ can be such that the inequality still holds.

\section*{Conclusions}

We have presented a proposal for a framework for market-consistent valuation of insurance liabilities which incorporates the Solvency II approach as a special case. We have shown that a value exists under certain conditions, derived upper bounds on such values, and shown that there exist a unique value defined by an explicit recursive expression if we restrict replication to risk-free zero-coupon bonds. The question remains open whether the value is unique in the general case.

Further, we have shown that the representation of the value as a sum of best estimate and risk margin is a simplification, which under certain conditions provides an upper bound on the value. By an explicit example, we have calculated the value as well as the upper bound, and have shown that the expression for the upper bound does not always preserve the order of the value.

\newpage

\bibliographystyle{plain}
\bibliography{MyBib}

\newpage

Author:\\
Christoph M\"ohr\\
Insurance Risk Management\\
Deloitte AG\\
General-Guisan-Quai 38\\
CH-8002 Zurich, Switzerland\\
Email: cmoehr@deloitte.ch

\end{document}